\documentclass[journal,twoside,web]{ieeecolor}

\usepackage{lcsys}
\usepackage{cite}
\usepackage{amsmath,amssymb,amsfonts}
\usepackage{graphicx}
\usepackage{hyperref}
\usepackage{flushend}
\let\labelindent\relax
\usepackage{enumitem}

\hypersetup{colorlinks=true,linkcolor=black,citecolor=blue}
\graphicspath{{Figs/}}

\newtheorem{theorem}{Theorem}
\newtheorem{proposition}{Proposition}
\newtheorem{definition}{Definition}
\newtheorem{remark}{Remark}

\catcode`\@=11
\def\downparenfill{$\m@th\braceld\leaders\vrule\hfill\bracerd$}
\def\overparen#1{\mathop{\vbox{\ialign{##\crcr\crcr \noalign{\kern0.4ex}
\downparenfill\crcr\noalign{\kern0.4ex\nointerlineskip}
$\hfil\displaystyle{#1}\hfil$\crcr}}}\limits}
\catcode`\@=12

\def\BibTeX{{\rm B\kern-.05em{\sc i\kern-.025em b}\kern-.08em
    T\kern-.1667em\lower.7ex\hbox{E}\kern-.125emX}}
\begin{document}
\title{Safety-Critical Stabilization of Force-Controlled Nonholonomic Mobile Robots}
\author{Tianyu Han and Bo Wang, \IEEEmembership{Member, IEEE}
\thanks{
This work was supported in part by GSoE at CCNY; and in part by PSC-CUNY Award, jointly funded by The Professional Staff Congress and The City University of New York. \textit{(Corresponding author: Bo Wang.)}}
\thanks{The authors are with the Department of Mechanical Engineering, The City College of New York, New York, NY 10031 USA (e-mail: than000@citymail.cuny.edu; bwang1@ccny.cuny.edu).}
}
\maketitle

\begin{abstract}
We present a safety-critical controller for the problem of stabilization for force-controlled nonholonomic mobile robots. The proposed control law is based on the constructions of control Lyapunov functions (CLFs) and control barrier functions (CBFs) for cascaded systems. To address nonholonomicity, we design the nominal controller that guarantees global asymptotic stability and local exponential stability for the closed-loop system in polar coordinates and construct a strict Lyapunov function valid on any compact sets. Furthermore, we present a procedure for constructing CBFs for cascaded systems, utilizing the CBF of the kinematic model through integrator backstepping. Quadratic programming is employed to combine CLFs and CBFs to integrate both stability and safety in the closed loop. The proposed control law is time-invariant, continuous along trajectories, and easy to implement. Our main results guarantee both safety and local asymptotic stability for the closed-loop system. 
\end{abstract}


\section{Introduction}\label{sec:introduction}

\IEEEPARstart{T}{he stabilization} problem of nonholonomic mobile robots has been thoroughly studied since the seminal work \cite{brockett1983asymptotic}, which established the non-existence of any smooth, time-invariant state feedback that stabilizes the origin of the closed-loop system. To address this challenge, researchers have developed various methods, including smooth time-varying feedback \cite{samson1995control,loria1999new}, discontinuous time-invariant feedback \cite{aicardi1995closed,astolfi1999exponential}, and hybrid feedback \cite{pomet1992hybrid}. While the literature on set-point stabilization for nonholonomic vehicles is now well-established, the body of work addressing 
stabilization with safety constraints is comparatively sparse and still has significant gaps.

In practical applications, autonomous systems must satisfy strict safety requirements, including the avoidance of obstacles and inter-vehicle collisions. As a result, ensuring safety is crucial in vehicle control applications. Various approaches have been developed in the literature to achieve safety requirements, \textit{e.g.}, the artificial potential field approach \cite{khatib1986real}, model predictive control \cite{lindqvist2020nonlinear}, behavior-based methods \cite{mataric1998behavior,wang2022robust}, and control barrier function (CBF)-based techniques \cite{ames2016control,wang2017safety,jankovic2018robust,jankovic2023multiagent,alan2023control,cortez2021robust}. Among the various methods proposed to achieve safety, the CBF-based approach is preferred in many applications due to its systematic enforcement of safety constraints, robustness against uncertainties, and ability to handle complex constraints effectively.

Using the CBF-based approach, the problem of multi-agent collision avoidance has been studied in \cite{wang2017safety,jankovic2023multiagent} for autonomous vehicles modeled by the double-integrator system. For vehicles modeled by the nonholonomic unicycle model, the safety-critical design presents additional challenges since the angular velocity does not appear through differentiation of the position, which results in limited control authority if the CBF is defined solely based on the vehicle's position \cite{glotfelter2019hybrid,haraldsen2023safety}. To address this challenge, the high-order CBF approach proposed in \cite{xiao2021high} applies to the unicycle model, as the control inputs appear through a second differentiation of the position. In \cite{taylor2022safe}, a safe backstepping procedure is proposed to construct (zeroing) CBFs for higher-order systems, which can handle higher-order systems with a mixed relative degree. In \cite{haraldsen2023safety}, an obstacle avoidance strategy is proposed for vehicles modeled by the nonholonomic integrator by regulating the vehicle speed and orientation separately via two CBFs, while maintaining nonzero forward speed in dynamic environments using velocity obstacles. Based on time-varying CBFs, a safety-critical controller is proposed for a kinematic unicycle model in \cite{huang2023obstacle} to achieve navigation and collision avoidance with both static and dynamic obstacles. An estimator-based safety-critical controller is presented in \cite{bohara2024safety} for formation tracking control of nonholonomic mobile robots, ensuring inter-agent collision avoidance using CBFs.

In this letter, we present a safety-critical controller designed for the problem of stabilization for force-controlled nonholonomic mobile robots. That is, our controller ensures \textit{asymptotic stability} for the closed-loop system while guaranteeing that the trajectories of the mobile robot remain within a predefined constrained configuration space at all times. The main contributions of this research include:
\begin{enumerate}[label={(\roman*)}]
    \item We construct a strict Lyapunov function that is valid on any compact sets for the nonholonomic robot model. This strict Lyapunov function originates from the design of a nominal controller in polar coordinates, which guarantees global asymptotic stability (GAS) and local exponential stability (LES) for the closed loop, serving as a control Lyapunov function (CLF) in the safety-critical design.

    \item A procedure for constructing (reciprocal) CBFs for cascaded systems is presented, utilizing the CBF of the kinematic model through integrator backstepping. 
\end{enumerate}\noindent
We employ a quadratic programming framework to combine CLFs and CBFs to integrate both \textit{stability} and \textit{safety} in the closed loop. The proposed control law is time-invariant, continuous along trajectories, and easy to implement. Our approach applies to scenarios such as autonomous parking with obstacle avoidance and inter-vehicle collision avoidance.

The structure of the remaining paper is as follows: In Section \ref{sec:CBF} we review the background of CBFs and safety-critical control and present an approach to construct CBFs for cascaded systems using integrator backstepping. Section \ref{sec:main} presents problem formulation and the main results on safety-critical stabilization for mobile robots. Section \ref{sec:simulation} presents simulation results illustrating the practical applications of our theoretical findings. Section \ref{sec:conclusion} offers concluding remarks.

\section{Constructions of CBFs}\label{sec:CBF}

\subsection{Control Barrier Functions and Safety-Critical Control}
Consider a nonlinear control-affine system
\begin{equation}
    \dot{x}=f(x)+g(x)u \label{eq:nls}
\end{equation}
where the state $x\in\mathbb{R}^n$ and the control $u\in\mathbb{R}^m$. We assume that $f$ and $g$ are locally Lipschitz and $f(0)=0$. The control objective is to design a state feedback $u=u^*(x)$ for system \eqref{eq:nls}, resulting in the closed-loop system 
\begin{equation}
    \dot{x}=f_{\text{cl}}(x):=f(x)+g(x)u^*(x) \label{eq:cl-nls}
\end{equation}
such that the trajectories of \eqref{eq:cl-nls} remain within the interior of a safe set $\mathcal{C}\subset \mathbb{R}^n$ at all times, and, if possible, render the origin of \eqref{eq:cl-nls} asymptotically stable. The safe set $\mathcal{C}$ is built as the 0-superlevel set of a continuously differentiable function $h:\mathbb{R}^n\to \mathbb{R}$, \textit{i.e.},
\begin{subequations}
  \begin{eqnarray}
    \mathcal{C}&=&\{x\in\mathbb{R}^n : h(x)\ge 0\}, \notag\\
    \partial\mathcal{C}&=&\{x\in\mathbb{R}^n : h(x)= 0\}, \notag\\
    \operatorname{int}\mathcal{C}&=&\{x\in\mathbb{R}^n : h(x)> 0\}. \notag
  \end{eqnarray}
\end{subequations}
The property of safety is formalized by requiring trajectories of \eqref{eq:cl-nls} to remain within the safe set $\mathcal{C}$ at all times.
\begin{definition}[Safety \cite{ames2019control}]\rm
    A set $\mathcal{C}$ is said to be \textit{forward invariant} for \eqref{eq:cl-nls} if for each initial condition $x_\circ\in \mathcal{C}$, the resulting trajectory $x:\mathbb{R}_{\ge 0}\to \mathbb{R}^n$ satisfies $x(t)\in \mathcal{C}$ for all $t\ge 0$. System \eqref{eq:cl-nls} is said to be \textit{safe} on a set $\mathcal{C}$ if it is forward invariant.
\end{definition}
For system \eqref{eq:nls}, we assume that we know a CLF that guarantees stability for the closed-loop system and a CBF that ensures adherence to constraints. The following definitions are standard.

\begin{definition}[Control Lyapunov Function (CLF)]\rm
    A positive definite, proper, differentiable function $V:\mathcal{D}\to \mathbb{R}_{\ge0}$ is a (local) \textit{CLF} for system \eqref{eq:nls} on $\mathcal{D}\subset \mathbb{R}^n$ if there exists a function $\alpha\in\mathcal{K}$ such that for all $x\in \mathcal{D}\backslash\{0\}$
    \begin{equation}\addtocounter{equation}{-1}
        L_gV(x)=0 \; \implies \; L_fV(x) + \alpha\left(|x|\right)<0.
    \end{equation}
\end{definition}
\vspace{2mm}

\begin{definition}[Control Barrier Function (CBF)]\rm
    A differentiable function $B:\operatorname{int}\mathcal{C}\to \mathbb{R}_{>0}$ is a (reciprocal) \textit{CBF} with respect to the admissible set $\mathcal{C}$ if $B(x)\to \infty$ as $x\to \partial \mathcal{C}$, and if there exists a function $\alpha_B\in\mathcal{K}$ such that for all $x\in \operatorname{int}\mathcal{C}$
    \begin{equation}
        L_gB(x)=0 \; \implies \; L_fB(x) - \alpha_B\left({1}/{B(x)}\right)<0. \label{eq:CBF}
    \end{equation}
\end{definition}

One effective method developed in \cite{jankovic2018robust} to combine a CLF and a CBF is based on the $\gamma m$-quadratic programming (QP):
\begin{align}
&\min ~ \frac{1}{2}(u^\top u +m \delta^\top\delta) \label{eq:QP}\\
\text{s.t.~~}&  \gamma_f(L_f V(x)+\alpha(|x|))+L_gV(x)u+L_gV(x)\delta\le 0 \notag\\
& L_f B(x) - \alpha_B\left({1}/{B(x)}\right) +L_gB(x)u \le 0   \notag
\end{align}
where $m\ge 1$, $\gamma_f$ is defined as $\gamma_f(s):=\gamma s$ if $s\ge 0$ and $\gamma_f(s):= s$ if $s<0$, and $\gamma\ge 1$.
The closed-form solution to the $\gamma m$-QP problem can be obtained by applying the Karush-Kuhn-Tucker (KKT) conditions, as described in \cite{jankovic2018robust}.

\subsection{Constructions of CBFs for Cascaded Systems}

Cascaded structures naturally arise in mechanical and robotic systems. For instance, in mechanical systems, the control force or torque enters only the kinetics subsystem, while the velocity variables impact the kinematics to control the configuration variables through interconnection terms. However, in general, the CBF of the kinematics subsystem usually cannot serve as a CBF for the entire cascades \cite{cohen2024safety}. In \cite{taylor2022safe}, a method to construct \textit{zeroing} CBFs for higher-order systems is proposed based on the CBFs of the reduced-order models. In this subsection, we propose a procedure for constructing \textit{reciprocal} CBFs for cascaded systems using integrator backstepping. While many previous works have used reciprocal CBFs to address the safety-stabilization problem for kinematic systems, our proposed method provides a convenient approach for extending these solutions to kinematic-kinetic cascaded structures.

Consider the following cascaded system
\begin{subequations} \label{eq:cascade}
  \begin{eqnarray}
    \dot{x}_1&=&f(x_1)+g(x_1)x_2\label{eq:cascade-a}\\
    \dot{x}_2&=&u\label{eq:cascade-b}
  \end{eqnarray}
\end{subequations}
where $x:=[x_1^\top~x_2^\top]^\top$, $x_1\in\mathbb{R}^n$, $x_2\in\mathbb{R}^m$, and $u\in\mathbb{R}^m$. Given an admissible set $\mathcal{C}\subset \mathbb{R}^n$, assume that we know a CBF $B_1(x_1)$ with respect to the admissible set $\mathcal{C}$ for system \eqref{eq:nls}\footnote{In general, $B_1(x_1)$ cannot serve as a CBF for the cascaded system \eqref{eq:cascade}.}.
Then there exists a ``virtual" safety-critical control law $x_2=x_2^*(x_1)$ for \eqref{eq:cascade-a} such that the subsystem 
\begin{equation}\label{eq:safe}
    \dot{x}_1=f_{\text{safe}}(x_1):=f(x_1)+g(x_1)x_2^*(x_1)
\end{equation}
is safe on the set $\mathcal{C}$\footnote{The safety-critical control law $x_2=x_2^*(x_1)$ may be given by the safety-critical filtering \cite{xu2015robustness} or given by solving the $\gamma m$-QP problem \eqref{eq:QP} \cite{jankovic2018robust}.}. That is, $B_1$ is a barrier function for \eqref{eq:safe}, \textit{i.e.}, $\dot{B}_1|_\eqref{eq:safe}<\alpha_{B}(1/B_1)$ for all $x\in \operatorname{int}\mathcal{C}$. Under the coordinate transformation $\tilde{x}_2:=x_2-x_2^*(x_1)$, the system \eqref{eq:cascade} becomes
\begin{subequations} \label{eq:cascade2}
  \begin{eqnarray}
    \dot{x}_1&=&f_{\text{safe}}(x_1)+g(x_1)\tilde{x}_2\label{eq:cascade2-a}\\
    \dot{\tilde{x}}_2&=&u - \dot{x}_2^*=:\tilde{u}.\label{eq:cascade2-b}
  \end{eqnarray}
\end{subequations}
We have the following result.
\begin{theorem}[Integrator backstepping]\label{thm:1}
    Consider the system \eqref{eq:cascade} and the admission set $\mathcal{C}\subset\mathbb{R}^n$. Suppose that we know a CBF $B_1:\operatorname{int}\mathcal{C}\to \mathbb{R}_{>0}$ for system \eqref{eq:nls} and a continuously differentiable ``virtual" controller $x_2^*:\mathbb{R}^n\to\mathbb{R}^m$ such that 
  \begin{align}
    L_{f_{\text{safe}}}B_1(x_1)&=L_fB_1(x_1)+L_g B_1(x_1)x_2^*(x_1)\notag\\
    &<\alpha_B(1/B_1(x_1))\label{eq:safecond}
  \end{align}
for some $\alpha_B\in\mathcal{K}$ and for all $x_1\in\operatorname{int}\mathcal{C}$. Then the function $B:\mathbb{R}^n\times\mathbb{R}^m\to \mathbb{R}_{>0}$ defined by    
    \begin{equation}
        B(\mathbf{x}):=B_1(x_1)+\tilde{x}_2^\top H \tilde{x}_2
    \end{equation}
    is a CBF for the system \eqref{eq:cascade} on the set $\mathcal{C}\subset\mathbb{R}^n\times\mathbb{R}^m$, where $\tilde{x}_2:=x_2-x_2^*$, the matrix $H=H^\top>0$, and $\mathbf{x}:=[x_1~\tilde{x}_2]^\top$.
\end{theorem}
\begin{proof}
    Note that the function $B(\mathbf{x})>0$ for all $\mathbf{x}\in\operatorname{int}\mathcal{C}$, where with a abuse of notation, $\operatorname{int}\mathcal{C}:=\{(x_1,x_2)\in\mathbb{R}^n\times\mathbb{R}^m:h(x_1)>0\}$. Furthermore, $B(\mathbf{x})\to \infty$ as $\mathbf{x}\to\partial \mathcal{C}:=\{(x_1,x_2)\in\mathbb{R}^n\times\mathbb{R}^m:h(x_1)=0\}$. Denoting
    \begin{equation*}
        F(\mathbf{x}):=\begin{bmatrix}
            f_{\text{safe}}(x_1)+g(x_1)\tilde{x}_2\\ 0
        \end{bmatrix}\;\;\text{and}\;\;
        G:=\begin{bmatrix}
            0 \\ I
        \end{bmatrix}
    \end{equation*}
    the system \eqref{eq:cascade2} can be written as 
    \begin{equation}
        \dot{\mathbf{x}}=F(\mathbf{x})+G\tilde{u}.
    \end{equation}
    Then we have
    \begin{equation}
        L_G B(\mathbf{x})=\frac{\partial \tilde{x}_2^\top H \tilde{x}_2}{\partial \tilde{x}_2}=2\tilde{x}_2^\top H,
    \end{equation}
    which implies that $L_G B(\mathbf{x})=0 \Longleftrightarrow \tilde{x}_2=0$ due to that the matrix $H$ is positive definite. Therefore, on the set $\{\tilde{x}_2=0\}$, we have
    \begin{equation}
        L_F B(\mathbf{x})|_{\tilde{x}_2=0}=\frac{\partial B_1(x_1)}{\partial x_1}f_{\text{safe}}(x_1)
    \end{equation}
    and 
    \begin{equation}
        \left.\alpha_B\left(\frac{1}{B(\mathbf{x})}\right)\right|_{\tilde{x}_2=0}=\alpha_B\left(\frac{1}{B_1(x_1)}\right).
    \end{equation}
    Finally, it follows from \eqref{eq:safecond} that 
    \begin{equation}
        L_G B(\mathbf{x})=0 \implies L_F B(\mathbf{x})-\alpha_B\left({1}/{B(\mathbf{x})}\right)<0,
    \end{equation}
    which verifies the condition \eqref{eq:CBF} and completes the proof.
\end{proof}

\section{Problem Formulation and Main Results}\label{sec:main}

Consider the force-controlled model of the nonholonomic mobile robots
\begin{align}
  &\left\{\begin{aligned}
    \dot{x}&=v\cos\theta\\
    \dot{y}&=v\sin\theta\\
    \dot{\theta}&=\omega
  \end{aligned}\right. \label{eq:WMR-1}\\ \vspace{4mm} 
  &\left\{\begin{aligned}
    \dot{v}&=u_1\\
    \dot{\omega}&=u_2
  \end{aligned}\right.\label{eq:WMR-2}
\end{align}
where $(x, y)$ represents the position of the robot in Cartesian coordinates; $\theta$ denotes the robot's orientation; $v$ and $\omega$ denote the forward and angular velocities, respectively; $u_1$ and $u_2$ denote translational and angular acceleration control inputs acting on the robots, respectively. Equation \eqref{eq:WMR-1} corresponds to the velocity kinematics model, while \eqref{eq:WMR-2} corresponds to the (simplified) force-balance kinetics equations.

The safety-critical stabilization problem involves rendering the origin of the closed-loop system asymptotically stable while ensuring that the trajectories of the closed-loop system remain within a predefined safe set $\mathcal{C}$ at all times $t\ge t_{\circ}$. 

As discussed in Section \ref{sec:CBF}, we address this problem using the CLF-CBF $\gamma m$-QP approach. Specifically, we construct the CLF and CBF for the nonholonomic mobile robot system separately. The control law is then derived from the $\gamma m$-QP \eqref{eq:QP}. However, to deal with the nonholonomicity, we perform the stabilization task and construct the CLF in polar coordinates, while to ensure safety, we construct the CBF in Cartesian coordinates.

\subsection{Stabilization and Constructions of CLFs}

Due to the nonholonomicity, directly finding a CLF for the mobile robot is not a trivial task. Therefore, we first design a stabilization controller for the robot in polar coordinates and then construct a strict Lyapunov function for the closed loop. Later, this strict Lyapunov function serves as the CLF in the safety-critical control design.

In polar coordinates, the position of the mobile robot may be represented by the distance $\rho$ and the bearing angle $\phi$, \textit{i.e.},
 \begin{equation}
     \rho = \sqrt{x^2 + y^2}, \quad \phi = \operatorname{atan2}(y, x).
 \end{equation}
As in \cite{aicardi1995closed,wang2021formation}, defining $\alpha:=\theta - \phi$, the kinematics equations \eqref{eq:WMR-1} become
\begin{align}
  &\left\{\begin{aligned}
    \dot{\rho} &= v \cos\alpha\\
   \dot{\phi} &= \frac{v}{\rho} \sin\alpha\\
    \dot{\alpha} &= \omega - \frac{v}{\rho} \sin\alpha.
  \end{aligned}\right. \label{eq:WMR-3}
\end{align}
We address the stabilization problem by separating the stabilization tasks at the kinematics and the kinetics levels. Specifically, we design virtual control laws $(v^*,\omega^*)$ to stabilize the origin for the error kinematics \eqref{eq:WMR-3}. Then, for the dynamics \eqref{eq:WMR-2}, we design a controller to ensure $(v,\omega)\to (v^*,\omega^*)$.

\subsubsection{Control of the error kinematics} 

It follows from \cite[Proposition 1]{wang2021formation} that the virtual control laws $(v^*,\omega^*)$ may be designed as 
\begin{subequations} \label{eq:virtual}
  \begin{eqnarray}
    v^*&=&-k_\rho\cos(\alpha)\rho\label{eq:virtual-a}\\
    \omega^*&=&-k_\alpha\alpha-k_\rho\operatorname{sinc}(2\alpha)(\alpha-\lambda\phi)\label{eq:virtual-b}
  \end{eqnarray}
\end{subequations}
where $k_\rho$, $k_\alpha$, and $\lambda$ are positive control gains, and the function $\operatorname{sinc}(\cdot)$ is defined as $\operatorname{sinc}(s):=\sin(s)/s$ if $s\ne 0$ and $\operatorname{sinc}(0):=1$. Note that the function $\operatorname{sinc}(\cdot)$ is smooth and bounded on $\mathbb{R}$, \textit{i.e.}, $\operatorname{sinc}(s)\in(-0.22,1]$ for all $s\in\mathbb{R}$. The closed-loop system with the virtual control laws $(v,\omega)= (v^*,\omega^*)$ becomes
\begin{subequations} \label{eq:errordyn}
  \begin{eqnarray}
    \dot{\rho} &=&  -k_\rho\cos({\alpha})^2\rho\label{eq:errordyn-a}\\
    \dot{\phi} &=&  -k_\rho\operatorname{sinc}(2\alpha)\alpha\label{eq:errordyn-b}\\
    \dot{\alpha} &=&  -k_{\alpha}\alpha+\lambda k_{\rho}\operatorname{sinc}(2\alpha)\phi.  \label{eq:errordyn-c}
  \end{eqnarray}
\end{subequations}
A \textit{weak} Lyapunov function for the system \eqref{eq:errordyn} is given by $V_1(\rho,\phi,\alpha):=\frac{1}{2}(\rho^2+\lambda \phi^2 + \alpha^2)$, which yields
\begin{equation}
    \dot{V_1} = -k_\rho \cos(\alpha)^2 \rho^2 - k_\alpha \alpha^2 \leq  0. \label{eq:V1}
\end{equation}
We conclude that the origin $(\rho,\phi,\alpha)=(0,0,0)$ is asymptotically stable for \eqref{eq:errordyn} following from Barbashin-Krasovskii-LaSalle’s invariance principle \cite[Proposition 1]{wang2021formation}. 

\subsubsection{Velocity tracking} Using $(v,\omega)=(v^*,\omega^*)+(\tilde{v},\tilde{\omega})$, the error kinematics become
\begin{equation}\label{eq:errorkinematics}
    \begin{bmatrix}
        \dot{\rho}\\\dot{\phi}\\ \dot{\alpha}
    \end{bmatrix}=
    \underbrace{\begin{bmatrix}
        -k_\rho\cos({\alpha})^2\rho\\
        -k_\rho\operatorname{sinc}(2\alpha)\alpha\\
         -k_{\alpha}\alpha+\lambda k_{\rho}\operatorname{sinc}(2\alpha)\phi
    \end{bmatrix}}_{f(\rho,\phi,\alpha)}
    +\underbrace{\begin{bmatrix}
        \rho\cos\alpha & 0\\
        \phantom{-}\sin\alpha & 0 \\
        -\sin\alpha & 1
    \end{bmatrix}}_{g(\rho,\alpha)}
    \begin{bmatrix}
        z \\ \tilde{\omega}
    \end{bmatrix}
\end{equation}
where $z:={\tilde{v}}/{\rho}$. Obviously, $(\tilde{v},\tilde{\omega})\to 0$ does not guarantee the boundedness of trajectories of \eqref{eq:errorkinematics} because $z={\tilde{v}}/{\rho}$ may approach infinity even when $\tilde{v}\to 0$. Therefore, we must analyze the error dynamics in the coordinates of $(z,\tilde{\omega})$, where
\vspace{-3mm}
\begin{subequations} \label{eq:errordyn2}
  \begin{eqnarray}
    \dot{z} &=&  \frac{u_1-\dot{v}^*}{\rho} + k_\rho \cos(\alpha)^2 z - \cos(\alpha)z^2 \label{eq:errordyn2-a}\\
    \dot{\tilde{\omega}} &=&  u_2-\dot{\omega}^*. \label{eq:errordyn2-b}
  \end{eqnarray}
\end{subequations}
Hence, choosing the control laws $(u_1,u_2)$ as
\begin{subequations} \label{eq:u}
  \begin{eqnarray}
    u_1 &=&  \dot{v}^* - \rho(k_\rho\cos({\alpha})^2 z - \cos({\alpha})z^2 + k_z z) \label{eq:u-a}\\
    u_2 &=&  \dot{\omega}^* - k_\omega\tilde{\omega}  \label{eq:u-b}
  \end{eqnarray}
\end{subequations}
yields the linear dynamics
\begin{subequations} \label{eq:errordyn3}
  \begin{eqnarray}
    \dot{z} &=&  -k_z z \label{eq:errordyn3-a}\\
    \dot{\tilde{\omega}} &=&  - k_\omega\tilde{\omega}. \label{eq:errordyn3-b}
  \end{eqnarray}
\end{subequations}
The closed-loop system \eqref{eq:errorkinematics}, \eqref{eq:errordyn3} has a cascaded structure, and the cascades argument \cite[Theorem 2]{panteley1998global} may be invoked to assess GAS of the origin. We have the following result.
\begin{proposition}[Nominal controller]\label{prop:1}
    Consider the mobile robot model \eqref{eq:WMR-3} and \eqref{eq:WMR-2} together with the control laws \eqref{eq:virtual} and \eqref{eq:u}, where $k_\rho$, $k_\alpha$, $k_z$, $k_\omega$, and $\lambda$ are positive constants. Then the origin of the closed-loop system \eqref{eq:errorkinematics}, \eqref{eq:errordyn3} is GAS.
\end{proposition}
\begin{proof}
    We regard \eqref{eq:errorkinematics} as a nominal system \eqref{eq:errordyn} with interconnection terms. The origin of the nominal system \eqref{eq:errordyn} is GAS, as follows from \eqref{eq:V1} and the invariance principle. The linear system \eqref{eq:errordyn3} is globally exponentially stable (GES). The interconnection matrix $g(\rho,\alpha)$ satisfies the condition of linear growth in $|(\rho,\phi,\alpha)|$. Hence, the origin of the cascaded system \eqref{eq:errorkinematics} and \eqref{eq:errordyn3} is GAS following from \cite[Theorem 2]{panteley1998global}.
\end{proof}
\begin{remark}
    The nominal controller \eqref{eq:virtual} and \eqref{eq:u} is smooth along trajectories and time-invariant, and it guarantees GAS for the closed-loop system in polar coordinates, i.e., GAS of the origin of \eqref{eq:errorkinematics}, \eqref{eq:errordyn3}. However, in Cartesian coordinates the nominal controller only guarantees the attractivity of the origin. Thus, Brockett’s necessary condition is not violated.
\end{remark}

\subsubsection{Constructions of strict Lyapunov functions}
Our objective is to construct a \textit{strict} Lyapunov function for the closed-loop cascaded system \eqref{eq:errorkinematics} and \eqref{eq:errordyn3} to serve as a CLF.
To this end, we first consider \eqref{eq:errordyn-b}-\eqref{eq:errordyn-c}. By adding and subtracting the term $k_\rho \alpha$ in \eqref{eq:errordyn-b}, and the term $\lambda k_\rho \phi$ in \eqref{eq:errordyn-c}, we obtain the output-injection form
\begin{equation}\label{eq:output-injection}
    \begin{bmatrix}
        \dot{\alpha}\\ \dot{\phi}
    \end{bmatrix}=
    \underbrace{\begin{bmatrix}
        -k_\alpha & \lambda k_\rho \\
        -k_\rho   & 0
    \end{bmatrix}}_{A}
    \begin{bmatrix}
        \alpha \\ \phi
    \end{bmatrix} + \underbrace{\begin{bmatrix}
        -\lambda k_\rho (1-\operatorname{sinc}(2\alpha))\phi)\\
        k_\rho(1-\operatorname{sinc}(2\alpha))\alpha
    \end{bmatrix}}_{K(\alpha,\phi)}.
\end{equation}
Noting that $\operatorname{sinc}(2\alpha)\to 1$ as $\alpha\to 0$, it follows that $K(\alpha,\phi)\to 0$ as the `output' $\alpha\to 0$.
The matrix $A$ is Hurwitz, and thus, there exists a matrix $P=P^\top>0$ such that 
\begin{equation}
    A^\top P +PA =-I.
\end{equation}
Consider the Lyapunov candidate 
\begin{equation}\label{eq:Lyap}
    V_2(\xi):=\xi^\top P \xi
\end{equation}
for \eqref{eq:errordyn-b}-\eqref{eq:errordyn-c}, where $\xi:=[\alpha~\phi]^\top$. Note that for any $r>0$ and for all $|\xi|<r$, we have $|K(\alpha,\phi)|\le c\;|\xi|\cdot|\alpha|$, where $c$ is a positive constant. Therefore, taking the time derivative of \eqref{eq:Lyap} yields
\begin{align}
    \dot{V}_2|_{\eqref{eq:output-injection}}&=-|\xi|^2 + 2 \xi^\top P K(\alpha,\phi) \notag\\
    &\le -|\xi|^2 + 2cr \lambda_M(P) |\xi|\cdot |\alpha| \notag\\
    &\le -|\xi|^2 + cr\lambda_M(P)\left(\frac{\alpha^2}{\varepsilon} + \varepsilon |\xi|^2\right)
\end{align}
where $\lambda_M(P)$ represents the maximum eigenvalue of $P$, and $\varepsilon>0$ can be chosen arbitrarily small. Choosing $\varepsilon:=\frac{1}{2cr\lambda_M(P)}$ yields
\begin{equation}
    \dot{V}_2|_{\eqref{eq:output-injection}}\le -\frac{1}{2}|\xi|^2 + 2c^2r^2 \lambda_M^2(P) \alpha^2,
\end{equation}
which is valid  for all $|\xi|\le r$ and for any $r>0$. We have the following results.
\begin{proposition}
The function 
\begin{equation}\label{eq:Lyap2}
    V(\rho,\phi,\alpha):=\nu(r)V_1(\rho,\phi,\alpha) + V_2(\xi)
\end{equation}
defined on $B_r:=\{|(\rho,\phi,\alpha)|\le r\}$ for any $r>0$ is a strict Lyapunov function (in the large) for \eqref{eq:errordyn}, where $\nu(r):=2c^2r^2 \lambda_M^2(P)/k_\alpha$.
\end{proposition}
\begin{proof}
    Taking the time derivative of \eqref{eq:Lyap2} yields
\begin{equation}\label{eq:34}
    \dot{V}|_{\eqref{eq:errordyn}}\le -\nu(r) k_\rho \cos(\alpha)^2 \rho^2 - \frac{1}{2}|\xi|^2,
\end{equation}
which implies that $\dot{V}$ is negative definite.
\end{proof}

\begin{proposition}
    For any $r>0$, there exists a constant $\bar{\mu}(r)>0$ such that for all $\mu\in(0,\bar{\mu}(r)]$, the function
    \begin{equation}
        \mathcal{V}_r(\rho,\phi,\alpha,z,\tilde{\omega}):=\mu\operatorname{ln}(V(\rho,\phi,\alpha)+1)+U(z,\tilde{\omega})
    \end{equation}
    defined on $B_r\times \mathbb{R}^2$ is a strict Lyapunov function (in the large) for the cascaded system \eqref{eq:errorkinematics} and \eqref{eq:errordyn3}, where 
    \begin{equation}
        U(z,\tilde{\omega}):=\frac{1}{2}\left(\frac{z^2}{k_z}+\frac{\tilde{\omega}^2}{k_\omega}\right).
    \end{equation}
\end{proposition}
\vspace{2mm}
\begin{proof}
    First, the function $\mathcal{V}_r$ is positive definite because $V$ is positive definite and $\mu>0$. The derivative of $\mathcal{V}_r$ along trajectories of \eqref{eq:errorkinematics} and \eqref{eq:errordyn3} is given by
    \begin{align}\label{eq:37}
        \dot{\mathcal{V}}_r=\frac{\mu L_fV}{V+1}+\frac{\mu L_gV}{V+1}  \cdot \zeta -|\zeta|^2
    \end{align}
    where $f$ and $g$ are defined in \eqref{eq:errorkinematics}, $\zeta:=[z~\tilde{\omega}]^\top$, and $L_fV$ is the same as $\dot{V}|_{\eqref{eq:errordyn}}$ in \eqref{eq:34}, which is negative definite in $(\rho,\phi,\alpha)$. The second term on the right-hand side of \eqref{eq:37} is indefinite, and following Young's inequality we have
    \begin{align}
        \frac{\mu L_gV}{V+1}  \cdot \zeta&\le \frac{1}{2}|\zeta|^2+\frac{\mu^2}{2}\frac{\left|L_gV\right|^2}{(V+1)^2}\notag  \\
        &\le \frac{1}{2}|\zeta|^2+\frac{\mu^2}{2}\frac{\left|L_gV\right|^2}{V+1}. \label{eq:38}
    \end{align}
    Since both $L_fV$ and $|L_gV|^2$ are quadratic functions in $(\rho,\phi,\alpha)$ in $B_r$, there exists a sufficiently small $\mu>0$ such that $\mu L_fV$ dominates $\frac{\mu^2}{2}|L_gV|^2$. Moreover, the term $-|\zeta|^2$ in \eqref{eq:37} dominates the term $\frac{1}{2}|\zeta|^2$ in \eqref{eq:38}.    
    Hence, $\dot{\mathcal{V}}_r$ is negative definite, which completes the proof.
\end{proof}

\subsection{Construction of the CBF}

Noting that the mobile robot \eqref{eq:WMR-1}-\eqref{eq:WMR-2} has a cascaded structure as in \eqref{eq:cascade}, the construction of the CBF for the mobile robot directly follows from Theorem \ref{thm:1}.

\begin{proposition}\label{prop:4}
Consider the system \eqref{eq:WMR-1}-\eqref{eq:WMR-2} and assume that the admissible set $\mathcal{C}:=\{(x,y)\in\mathbb{R}^2:h(x,y)\ge 0\}$ is given, where $h:\mathbb{R}^2\to \mathbb{R}_{>0}$ is continuously differentiable.
\begin{enumerate}[label={(\roman*)}]
    \item For the kinematics system \eqref{eq:WMR-1}, if we consider $(v,\omega)$ as the input, then $B_1(x,y):=1/h(x,y)$ is a CBF for \eqref{eq:WMR-1}. 

    \item For the cascaded system \eqref{eq:WMR-1}-\eqref{eq:WMR-2}, the function $B(\mathbf{x}):=B_1(x,y)+\eta^\top H\eta$ is a CBF, where $\mathbf{x}:=[x~y~v~\omega]^\top$, $\eta:=[v~\omega]^\top$, and $H=H^\top >0$.
\end{enumerate}
\end{proposition}
\begin{proof}
(i)  The function $B_1$ is continuously differentiable in $\operatorname{int}\mathcal{C}$ because $h(x,y)>0$ for all $(x,y)\in\operatorname{int}\mathcal{C}$. Furthermore, $B_1\to \infty$ as $(x,y)\to \partial \mathcal{C}$. Moreover, for all $(x,y)\in \operatorname{int}\mathcal{C}$ and any $\alpha_B\in \mathcal{K}$, we have $\alpha_B(1/B_1(x,y))=\alpha_B(h(x,y))>0$. Due to that the kinematics system \eqref{eq:WMR-1} is driftless (\textit{i.e.}, $f\equiv 0$), the term $L_fB_1(x)$ in \eqref{eq:CBF} is zero, and $L_fB_1(x,y)-\alpha_B(1/B_1(x,y))<0$ for all $x\in\operatorname{int}\mathcal{C}$. Thus, we verify the condition \eqref{eq:CBF}.

Item (ii) follows directly from Theorem \ref{thm:1} by noticing that for the kinematics system \eqref{eq:WMR-1} the velocity control law $(v,\omega)\equiv(0,0)$ is a virtual safety-critical control law.
\end{proof}

\subsection{Safety-Critical Control Design}

We have shown that $\mathcal{V}_r$ is a strict Lyapunov function for the closed-loop system \eqref{eq:errorkinematics} and \eqref{eq:errordyn3}, and thus, it serves as a CLF for system \eqref{eq:WMR-3} and \eqref{eq:WMR-2}. Furthermore, according to Proposition \ref{prop:4}, $B(\mathbf{x}):=B_1(x,y)+\eta^\top H \eta$ is a CBF for \eqref{eq:WMR-1}-\eqref{eq:WMR-2}.
By combining both the stability and safety constraints, the safety-critical stabilization control law is derived by solving the $\gamma m$-QP problem. 

For the sake of concision, we define $\bar{u}:=[\frac{u_1}{\rho},u_2]^\top$, $f_{\kappa}(\rho,\phi,\alpha):=[v\cos(\alpha)~\frac{v}{\rho}\sin\alpha~\omega-\frac{v}{\rho}\sin\alpha]^\top$,
\begin{align*}
    f_1&:=\begin{bmatrix}
        f_{\kappa}(\rho,\phi,\alpha)\\ 
        -\frac{\dot{v}^*}{\rho}+k_\rho\cos(\alpha)^2z - \cos(\alpha)z^2\\
        -\dot{\omega}^*
    \end{bmatrix},\quad
    g_1:=\begin{bmatrix}
        0_{3\times 2} \\ I_{2}
    \end{bmatrix},\\
    f_2&:=\begin{bmatrix}
        v\cos\theta ~ v\sin\theta ~\omega ~ 0 ~ 0
    \end{bmatrix}^\top,\quad
    g_2:=\begin{bmatrix}
        0_{3\times 2} \\ \operatorname{diag}(\rho,1)
    \end{bmatrix}.
\end{align*}
Then, for any $r>0$, the $\gamma m$-QP problem is defined as
\begin{align}
&\min ~ \frac{1}{2}(\bar{u}^\top \bar{u} +m \delta^\top\delta) \label{eq:QP1}\\
\text{s.t.~~}& F_1:=  \gamma_f(L_{f_1}\mathcal{V}_r+\alpha(|\chi|))+L_{g_1}\mathcal{V}_r\bar{u}+L_{g_1}\mathcal{V}_r\delta\le 0 \notag\\
& F_2:=L_{f_2} B(\mathbf{x}) - \alpha_B\left({1}/{B(\mathbf{x})}\right) +L_{g_2}B(\mathbf{x})\bar{u} \le 0   \notag
\end{align}
where $\chi:=[\rho~\phi~\alpha~z~\tilde{\omega}]^\top$, $\alpha_B\in\mathcal{K}$, $\alpha:=\frac{\epsilon L_fV}{V+1}-\frac{1}{2}|\zeta|^2$, and $\epsilon>0$ is chosen to be sufficiently small. The closed-form solution of the $\gamma m$-QP problem \eqref{eq:QP1} can be obtained by invoking the KKT conditions. We refer the readers to \cite{jankovic2018robust}.

\begin{proposition}\label{prop:5}
The $\gamma m$-QP problem \eqref{eq:QP1} is feasible, and under the resulting control law the set $\operatorname{int}\mathcal{C}$ is forward invariant. If $0\in\operatorname{int}\mathcal{C}$, then the barrier constraint is inactive ($F_2<0$) around the origin, and the resulting control law is continuous. If we select $\frac{\gamma m}{m+1}=1$, the origin of the closed-loop system is locally asymptotically stable.   
\end{proposition}

We omit the proof since it is closely parallel to that of \cite[Theorem 1]{jankovic2018robust}. It should be noted that the $\gamma m$-QP design \eqref{eq:QP1} may lead to undesirable equilibria because it prioritizes keeping the admissible set $\mathcal{C}$ invariant while attempting to achieve stabilization if possible \cite{jankovic2018robust}.

\section{Simulation Results}\label{sec:simulation}

We present two simulation examples to illustrate the performance of the proposed controller. In each example, the proposed CLF-CBF $\gamma m$-QP controller is compared with the nominal controller \eqref{eq:virtual} and \eqref{eq:u} from Proposition \ref{prop:1} and the CLF-based pointwise minimum norm (PMN) controller \cite{freeman2008robust}.

\subsubsection*{Example 1} Suppose that we want to stabilize the nonholonomic mobile robot to the origin while staying within the admissible set defined by $\mathcal{C}:=\{(x,y)\in\mathbb{R}^2:h(x,y)=1 + x -8y^2>0\}$. To this end, we define the CBF $B(\mathbf{x}):=B_1(x,y)+0.1(v^2+\omega^2)$
where $\mathbf{x}:=[x~y~v~w]^\top$ and $B_1(x,y):={1}/{h(x,y)}$ with $\alpha_B(s):=s$. The initial conditions of the mobile robot are randomly chosen as $(x,y,\theta,v,w)(0)=(7,0.63,2.55,-3.73,4.13)$. We select control parameters as $\lambda=3$, $k_\rho=2$, $k_\alpha=4$, $k_z=6$, $k_\omega=6$, $\nu=10$, $\mu=1$, $m = 2$, and $\gamma=1.5$.  All parameters are given in SI units. The simulation results are shown in Figs. \ref{fig:path} and \ref{fig:error}, which depict the physical path and configuration trajectories of the robot.

\begin{figure}[t]
    \centering
    \includegraphics[scale=0.3]{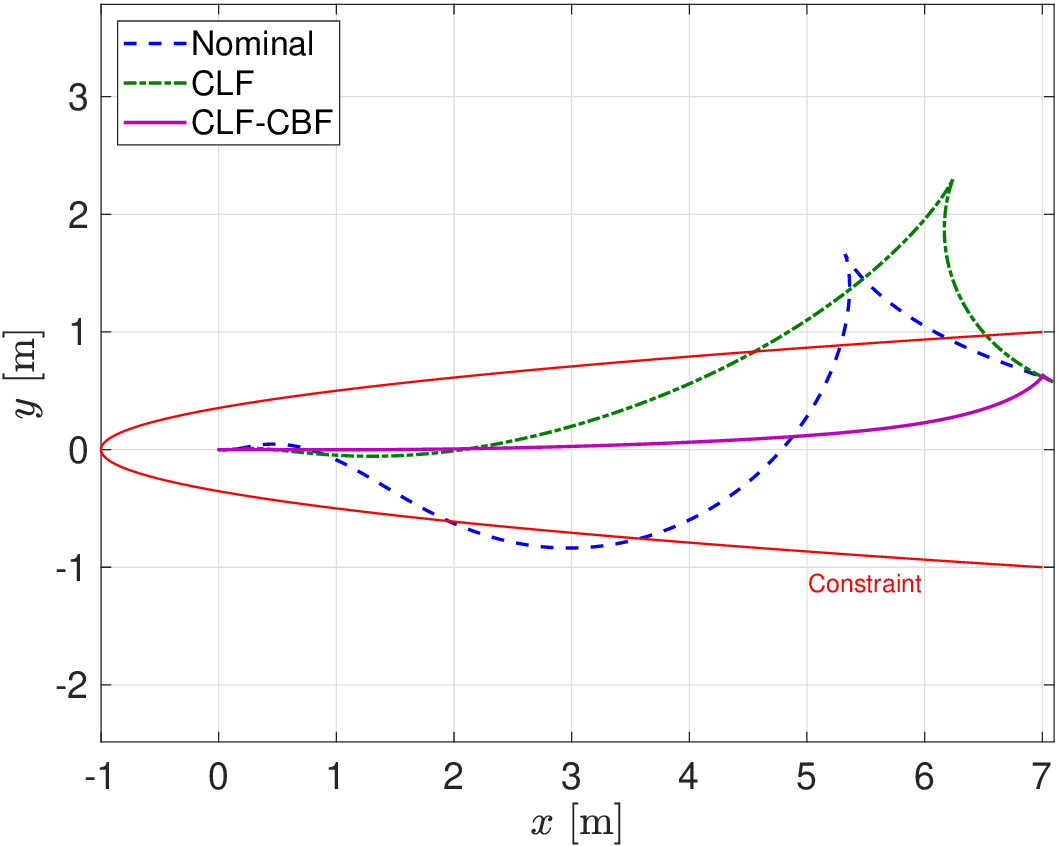}
    \caption{Illustration of the mobile robot paths in stabilization (Example 1).}
    \label{fig:path}
\end{figure}
\begin{figure}[t]
    \centering
    \includegraphics[scale=0.3]{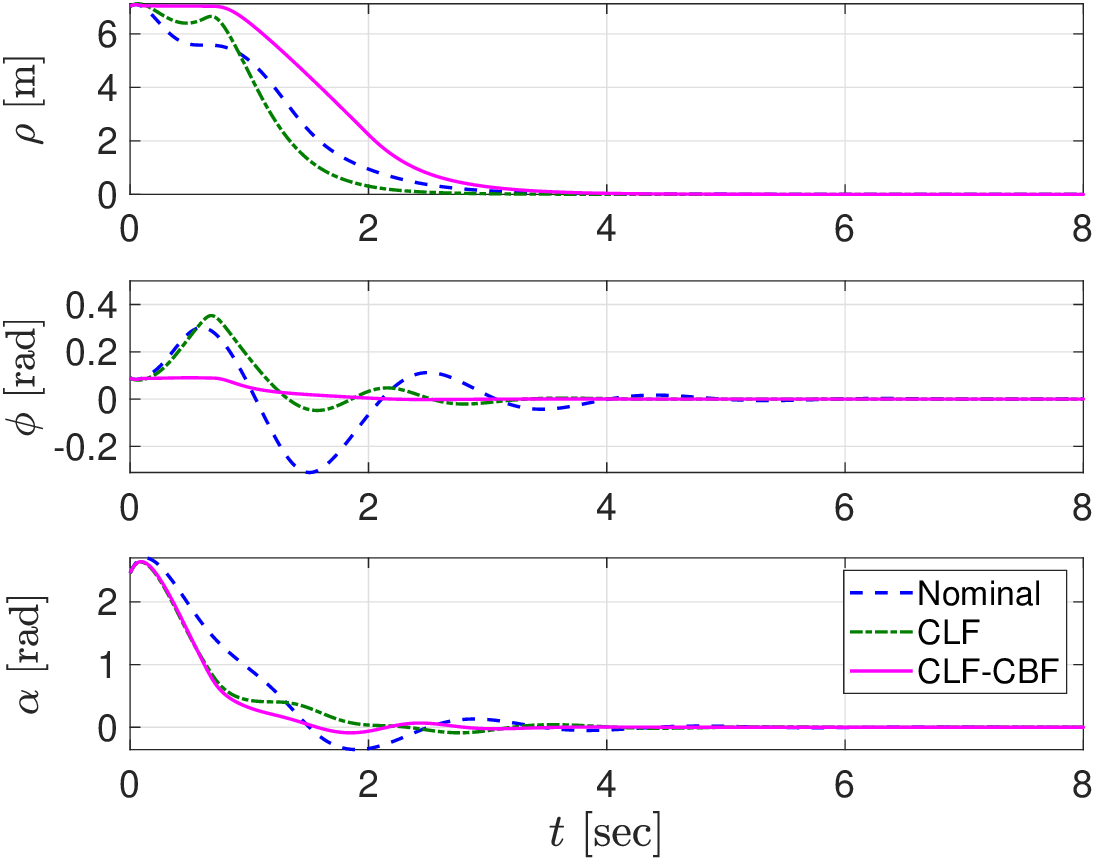}
    \caption{Convergence of the configuration variables of the mobile robot in polar coordinates (Example 1).}
    \label{fig:error}
\end{figure}

\subsubsection*{Example 2} We assume that the admissible set is given by $\mathcal{C}:=\{(x,y)\in\mathbb{R}^2:h(x,y)=1 - x + 8y^2>0\}$. We define the CBF $B(\mathbf{x}):=B_1(x,y)+v^2+\omega^2$ with $\alpha_B(s):=s$, where $B_1(x,y):={1}/{h(x,y)}$. The initial conditions are randomly chosen as $(x,y,\theta,v,w)(0)=(-5,4.58,4.65,-3.42,4.71)$. The control parameters are set to be the same as in Example 1.  The simulation results are shown in Figs. \ref{fig:path2} and \ref{fig:error2}, which demonstrate that the proposed CLF-CBF $\gamma m$-QP controller effectively achieves parking with obstacle avoidance.

\begin{figure}[t]
    \centering
    \includegraphics[scale=0.3]{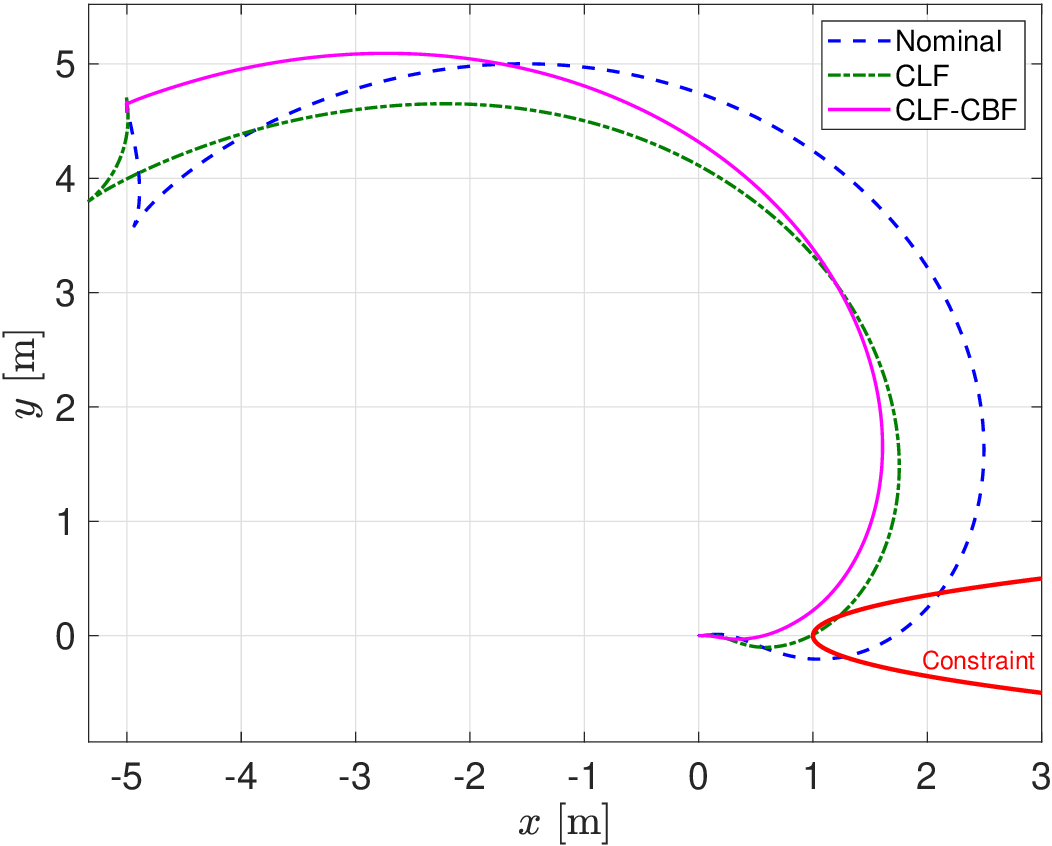}
    \caption{Illustration of the mobile robot paths in stabilization (Example 2).}
    \label{fig:path2}
\end{figure}
\begin{figure}[t]
    \centering
    \includegraphics[scale=0.3]{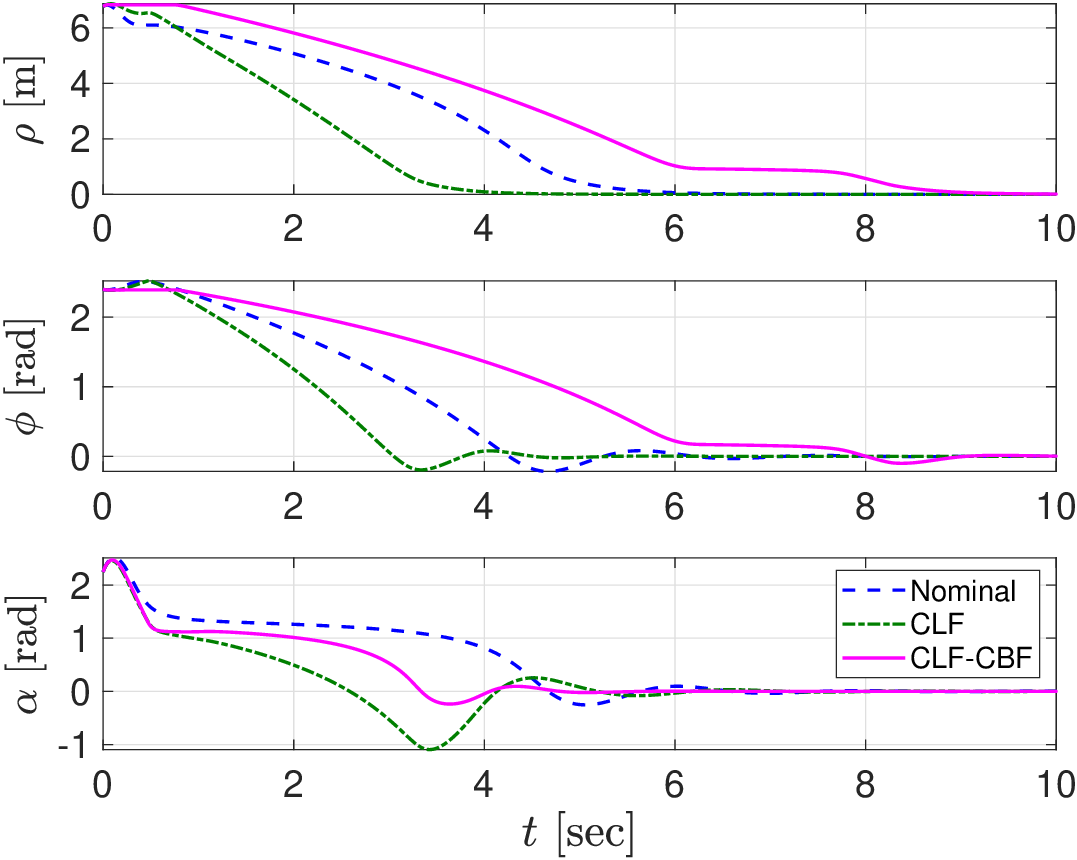}
    \caption{Convergence of the configuration variables of the mobile robot in polar coordinates (Example 2).}
    \label{fig:error2}
\end{figure}

\section{Conclusion}\label{sec:conclusion}

In this letter, we solve the stabilization problem with guaranteed safety for force-controlled nonholonomic mobile robots. Our main contributions lie in the construction of a strict Lyapunov function that is valid on any compact sets for the nonholonomic robot model in polar coordinate, serving as a CLF in the safety-critical stabilization design, and the construction of reciprocal CBFs for (kinematics-kinetics) cascaded systems, utilizing the CBF of the kinematic model through integrator backstepping. Quadratic programming is employed to integrate both stability and safety in the closed loop. 
Future research will focus on addressing input constraints, which are common in robotic systems. For instance, the approach presented in \cite{cortez2021robust} could be integrated with the method proposed in this paper to address the safety-stabilization problem for mobile robots under input saturation. Additionally, we plan to extend this work to safety formation control for multi-agent systems.

\bibliographystyle{ieeetr} 
\bibliography{mybibfile}   

\end{document}